%% file: Paper.tex
\documentclass{llncs}
\usepackage{amsmath}
\usepackage{amssymb}
\usepackage{algorithmic}
\usepackage{algorithm}
\usepackage{color}
\usepackage{times}
\usepackage{tikz}

\let\oldparagraph\paragraph
\renewcommand{\paragraph}[1]{\oldparagraph{\bf{#1}}}

\newcommand{\mi}{_{-i}}
\newcommand{\mj}{_{-j}}
\newcommand{\st}{ \:|\: }

\newcommand{\su}{\subseteq}
\newcommand{\E}{\mathbf E}
\newcommand{\R}{\mathbb R}
\allowdisplaybreaks

\begin{document}

\title{Bounding the Inefficiency of Altruism Through Social Contribution Games}

\author{%
Mona Rahn\inst{1}\and
Guido Sch\"afer\inst{1,2}
}

\institute{
CWI and VU University Amsterdam, The Netherlands \\
\email{\{rahn,g.schaefer\}@cwi.nl}
}

\institute{
CWI Amsterdam, The Netherlands 
\and
VU University Amsterdam, The Netherlands \\
\email{\{rahn,g.schaefer\}@cwi.nl}
}

\maketitle

\input{abstract.tex}

\section{Introduction}

\input{intro.tex}

\sloppy

\section{Preliminaries}
Let $G =  (N, \{\Sigma_i\}_{i \in N}, \{C_i\}_{i \in N})$ be a \emph{cost-minimization game}, where $N$ is the set of players, $\Sigma_i$ is player $i$'s strategy space, $\Sigma = \prod_{i \in N} \Sigma_i$ is the set of strategy profiles, and $C_i: \Sigma \to \mathbb R$ denotes the cost player $i$ must pay for a given strategy profile. We assume that each player seeks to minimize his cost. A \emph{social cost function} $C: \Sigma \to \R$ assigns a social cost to each strategy profile. We usually require $C$ to be \emph{sum-bounded}, i.e., $C(s) \leq \sum_{i \in N} C_i(s)$ for all $s \in \Sigma$.

We denote \emph{payoff-maximization games} as $G = (N, \{\Sigma_i\}_{i \in N}, \{\Pi_i\}_{i \in N})$ with \emph{social welfare} $\Pi: \Sigma \to \R$. In this case, each player $i$ tries to maximize his \emph{utility (or payoff)} $\Pi_i$. Again, we usually assume that $\Pi$ is \emph{sum-bounded}, i.e. $\Pi(s) \geq \sum_{i \in N} \Pi_i(s)$ for all $s \in \Sigma.$

In the following, we state most of the definitions and theorems only for cost-minimization games. The payoff-maximization case works  similarly by reversing all inequalities.

\begin{definition}
A \emph{coarse equilibrium} of a cost-minimization game $G$ is a probability distribution $\sigma$ over $\Sigma$ such that the following holds:
If $s$ is a random variable with distribution $\sigma$, then for all players $i$ and all strategies $s_i^* \in \Sigma_i$, $\E_{s \sim \sigma} [C_i(s)] \leq \E_{s\mi \sim \sigma\mi} [C_i(s_i^*, s\mi)]$,
where $\sigma\mi$ is the projection of $\sigma$ on $\Sigma\mi = \Sigma_1 \times \ldots \times \Sigma_{i-1} \times \Sigma_{i+1} \times \ldots \times \Sigma_n.$ A \emph{mixed Nash equilibrium} is a coarse equilibrium $\sigma$ that is the product of independent probability distributions $\sigma_i$ on $\Sigma_i$ (for $i \in N$). A \emph{(pure) Nash equilibrium} is a strategy profile $s \in \Sigma$ such that for all $s^* \in \Sigma$, $C_i(s) \leq C_i(s_i^*, s\mi)$, where $s\mi$ denotes the restriction of $s$ to $\Sigma\mi$.

The \emph{coarse (resp. correlated, mixed, pure) price of anarchy (PoA)} of a cost-mini\-mization game $G$ is defined as $\sup_s C(s)/C(s^*)$, where $s^*$ minimizes $C$ and $s$ runs over the coarse (resp. correlated, mixed, pure) Nash equilibria of $G$.\footnote{Similarly, we define the respective types of PoA for a payoff-maximization game as $\sup \Pi(s^*) / \Pi(s)$, where $s$ and $s^*$ are as above.} The coarse (resp. correlated, mixed, pure) PoA of a class $\mathcal G$ of games is defined as the supremum of the respective PoA values of games in $\mathcal G$.
\end{definition}

Note that pure Nash equilibria constitute a subset of mixed Nash equilibria which constitute a subset of coarse equilibria.
This implies that that the respective prices of anarchy are non-decreasing (in this order).

\subsection{The Altruism Model}
\begin{definition}[\cite{altruistic_games}]
Let $\alpha \in [0,1]^N$ and $G$ be a cost-minimization game. The \emph{$\alpha$-altruistic extension} of $G$ is defined as the cost-minimization game $G^\alpha = (N, \{\Sigma_i\}_{i \in N}, \{C^\alpha_i\}_{i \in N})$, where for any $i \in N$, the \emph{perceived cost} is the convex combination
$
C_i^\alpha = (1-\alpha_i)C_i + \alpha_i C
$
(interpreting $\R^\Sigma$ as a real vector space).
$G$ is called the \emph{base game}. The social cost function of $G^\alpha$ is again $C$, i.e., the cost of the base game.
\end{definition}

The higher the `altruism level' $\alpha_i$ is, the more $i$ cares about the society in general: player $i$ behaves egoistically if $\alpha_i = 0$, whereas he is completely altruistic if $\alpha_i = 1$.

\begin{definition}
Let $G$ be a cost-minimization game with sumbounded social cost and consider an altruistic extension $G^\alpha$ of $G$. $G^\alpha$ is $(\lambda, \mu)$\emph{-smooth} if there exists an optimal strategy $s^*$ such that for any strategy $s \in \Sigma$,
\[
\sum_{i \in N} \big (C_i(s_i^*, s\mi) + \alpha_i(C\mi (s_i^*, s\mi) - C\mi(s)) \big ) \leq \lambda C(s^*) + \mu C(s),
\]
where we abbreviate $C\mi := C - C_i$. 

The \emph{robust PoA} of $G^\alpha$ is defined as $\inf\{ \frac{\lambda}{1-\mu} \st \text{$G^\alpha$ is $(\lambda, \mu)$-smooth}, \,\mu<1\}$.
\end{definition}

\begin{theorem}[{\cite{altruistic_games}}]
Let $G^\alpha$ be an $\alpha$-altruistic extension of $G$. Then the coarse (and thus the correlated, mixed and pure) PoA of $G^\alpha$ is bounded from above by the robust PoA of $G^\alpha.$
\end{theorem}

\subsection{The Friendship Model}
\begin{definition}[\cite{inefficiency}]
Let $G = (N,\{\Sigma_i\}_{i \in N}, \{C_i\}_{i \in N})$ be a cost-minimization game with social cost $C$ and $\alpha \in [0,1]^{N \times N}$ such that $\alpha_{ii} = 1$ for all $i \in N$. The $\alpha$-\emph{friendship extension} of $G$ is defined as $G^\alpha = (N,\{\Sigma_i\}_{i \in N}, \{C^\alpha_i\}_{i \in N})$, where for any $i \in N$, the \emph{perceived cost} is defined as
$
C^\alpha_i = \sum_j \alpha_{ij} C_j.
$
Like in the altruism model, we consider $C$, the social cost function of the base game, as the social cost for $G^\alpha$.
\end{definition}

For players $i$ and $j$, $\alpha_{ij}$ can be interpreted as the level of affection $i$ feels towards $j$. Note that if $C = \sum_j C_j$, then the altruism model is a special case of the friendship model because in this case, $C_i^\alpha = C_i + \sum_{j \neq i} \alpha_i C_j.$

Next we adapt the smoothness definition in \cite{inefficiency} for the friendship model to the weighted player case; we will later need this to bound the robust PoA for weighted completion time scheduling games.

\begin{definition} 
Let $G^\alpha$ be friendship extension of a cost-minimization game with a \emph{weight-bounded} social cost function, i.e., $C \leq \sum_i w_i C_i$ for some $w \in \mathbb R^N_+$. $G^\alpha$ is $(\lambda, \mu)$-\emph{smooth} if there exists a (possibly randomized) strategy profile $\bar s$ such that for all strategy profiles $s$ and all optima $s^*$, 
\[
\sum_i w_i \big(C_i(\bar s_i, s\mi) + \sum_{j \neq i} \alpha_{ij} (C_j(\bar s_i, s\mi)-C_j(s)) \big) 
\leq  \lambda C(s^*) + \mu C(s).
\]
We define the \emph{robust PoA} of $G^\alpha$ as $\inf\{ \frac{\lambda}{1-\mu} \st \text{$G^\alpha$ is $(\lambda, \mu)$-smooth}, \, \mu<1\}$
\end{definition}

\begin{theorem}\label{weighted_smoothness_thm}
Let $G^\alpha$ be a friendship extension of a cost-minimization game with weight-bounded social cost function $C$. If $G^\alpha$ is $(\lambda, \mu)$-smooth with $\mu < 1$, then the coarse PoA of $G^\alpha$ is at most $\frac{\lambda}{1-\mu}$.
\end{theorem} 
The proof can be found in Appendix \ref{weighted_smoothness_appendix}.
\newline

One can also generalize the smoothness definition of the altruism model to weighted social costs and by allowing arbitrary $\bar s$ instead of the optimal $s^*$ in the term that is to be bounded. However, we do not need such generality in this paper and thus leave it out for simplicity.

\section{Social Contribution Games}
\begin{definition}
Let $G = (N, \{\Sigma_i\}_{i \in N}, \{C_i\}_{i \in N})$ be a cost-minimization game with social cost $C: \Sigma \to \mathbb R$. We call $G$ a \emph{(cost-minimization) social contribution game (SCG)} if for all players $i$ there exists a default strategy $\emptyset_i$ such that for all $s \in \Sigma,$
\[
C_i(s) = C(s) - C(\emptyset_i, s\mi).
\]
\end{definition}

The strategy $\emptyset_i$ is often interpreted as `refusing to participate in the game'. In that sense, $i$ pays exactly the social cost he causes by choosing to play; in the payoff-maximization case, he gets exactly what he contributes to the social welfare. So social contribution games are `fair' in some sense.

Basic utility games \cite{vetta} satisfy the definition of an SCG (see also Section \ref{section_valid_utility}). In particular, the competitive facility location game (which is a basic utility game by \cite{vetta}) is an SCG.

We now show that social contribution games satisfy a nice invariance property with respect to their $\alpha$-altruistic extensions.

\begin{lemma}
Any social contribution game is \emph{altruism-independently smooth}, i.e., for all $\alpha = (\alpha_i)_{i \in N}$ and corresponding altruistic extensions $G^\alpha$ of $G$, the robust price of anarchy in $G$ and $G^\alpha$ is the same.
\end{lemma}

\begin{proof}
For all players $i$, $C \mi (s) = C(s) - C_i(s)$ is independent of $s_i$ since $C(s) - C_i(s) = C(\emptyset_i, s\mi)$. Thus for all strategy profiles $s, s^*$, and all $\alpha \in \mathbb R^N$,
\[
\sum_i \big( C_i (s_i^*, s \mi) + \alpha_i ( C \mi(s_i^*, s\mi) - C \mi (s) )\big) = \sum_i C_i (s_i^*, s \mi).
\]
It follows that for all $(\lambda, \mu) \in \mathbb R^2$, $G^\alpha$ is $(\lambda, \mu)$-smooth iff $G$ is.
\qed
\end{proof}

The notions of $\alpha$-altruistic extensions and $\alpha$-independent smoothness can be easily extended to $\alpha \in \mathbb R^N$. The above lemma continues to hold in this case. So even if a player wants to \emph{hurt} society, the robust PoA stays the same.

\subsection{Social Contribution Bounded Games}
\begin{definition}
Let $G = (N, \{\Sigma_i\}_{i \in N}, \{C_i\}_{i \in N})$ a cost-minimization game with sumbounded social cost $C: \Sigma \to \mathbb R$. We call $G$ \emph{social contribution bounded (SC-bounded)} if for all players $i$ there exists a default strategy $\emptyset_i$ such that for all $s \in \Sigma,$
\[
C_i(s) \leq C(s) - C(\emptyset_i, s\mi).
\]

In this case, we define \emph{the corresponding social contribution game} $\bar G = (N, \{\Sigma_i\}_{i \in N}, \{\bar C_i\}_{i \in N})$ by setting 
$
\bar C_i (s) = C(s) - C(\emptyset_i, s\mi).
$
\end{definition}

Again, we think of $\emptyset_i$ as the option that $i$ does not participate. Note that $\emptyset_i$ need not actually be an element of $\Sigma_i$. In many games such as scheduling or congestion games, it is not an option not to participate (i.e., not to use any resources). So formally, we should require: There exists a function $\mathfrak C: \prod_{i \in N} (\Sigma_i \cup \{\emptyset_i\}) \to \mathbb R$ such that $\mathfrak C |_{\Sigma} = C$ and $C_i(s) \leq \mathfrak C(s) - \mathfrak C(\emptyset_i, s\mi)$ for all $i$ and $s$. However, there is a natural way to extend $C$ (and $C_i$) on $\prod_{i \in N} (\Sigma_i \cup \{\emptyset_i\})$, as we will see later. So for simplicity of notation, we write $C$ instead of $\mathfrak C$.

The following theorem shows that if we want to get a bound on the PoA of $\alpha$-altruistic extensions of an SC-bounded game, we might as well consider the corresponding SCG \emph{regardless of $\alpha$}.

\begin{theorem}\label{corresp_altruism}
Let $G$ be social contribution bounded and suppose that the robust price of anarchy of the corresponding SCG $\bar G$ is $\xi$. Then for all altruistic extensions $G^\alpha$ of $G$, the robust price of anarchy is at most $\xi$.
\end{theorem}

\begin{proof}
Let $s, s^*\in \Sigma$. We know that for $\alpha_i \in [0,1]^N$,
\begin{align*}
& C_i(s_i^*, s\mi) + \alpha_i (C \mi (s_i^*, s\mi) - C \mi (s)) \\
& \qquad = 
(1- \alpha_i) C_i(s_i^*, s\mi) + \alpha_i (C (s_i^*, s\mi) - C \mi (s)) \\
& \qquad \leq (1-\alpha_i) (C(s_i^*, s \mi) - C(\emptyset_i, s\mi)) 
+ \alpha_i (C (s_i^*, s\mi) - C(\emptyset_i, s\mi)) \\
& \qquad = \bar C_i (s_i^*, s\mi),
\end{align*}
where the inequality follows from applying SC-boundedness twice.
Summing over all players $i$, it follows that $G^\alpha$ is $(\lambda, \mu)$-smooth if $\bar G$ is.
\qed
\end{proof}

Now, in order to be able to make statements about \emph{friendship} extensions, we need a slightly stronger definition.

\begin{definition}
Assume a cost minimization game $G$ with weight-bounded social cost satisfies three assumptions for all $s \in \Sigma$ and players $i$: 
\begin{enumerate}
  \setlength{\parskip}{0pt}
\item $\: C_i(\emptyset_i, s\mi) = 0$ (if $i$ does not participate, he pays nothing)
\item $\forall j \neq i: \:C_j(\emptyset_i, s\mi) \leq C_j(s)$ (other players' costs can only increase if $i$ participates)
\item $w_i \sum_j (C_j (s) - C_j(\emptyset_i, s\mi)) \leq C(s) - C(\emptyset_i, s\mi)$ (the weighted impact of $i$'s participation on the players' costs is bounded by his impact on the social cost)
\end{enumerate} 
Then we call $G$ \emph{strongly SC-bounded}. 
\end{definition}

If all weights are 1, then assumption (3) easily follows from \newline
\emph{3b. $C(s) = \sum_j C_j(s)$ (social cost is sum of individual costs).}\newline

Using this definition, we are able to derive bounds on friendship extensions:

\begin{theorem}\label{corresp_relationship}
Let $G$ be strongly SC-bounded. Suppose the robust price of anarchy of $\bar G$ is $\xi$. Then for all friendship extensions $G^\alpha$, the robust price of anarchy is at most $\xi$.
\end{theorem}

\begin{proof}
Consider the friendship extension $G^\alpha$ of $G$, where $C_i^\alpha = \sum_j \alpha_{ij} C_j$, $\alpha_{ij} \in [0,1]$, $\alpha_{ii} = 1$. We calculate that for all $i$:
\begin{align*}
& w_i\big(C_i(\bar s_i, s\mi) + \sum_{j\neq i} \alpha_{ij} (C_j(\bar s_i, s\mi) - C_j(s))\big)  \\
&\qquad \overset{(2)}{\leq} w_i\big(C_i(\bar s_i, s\mi) + \sum_{j\neq i} \alpha_{ij} (C_j(\bar s_i, s\mi) - C_j(\emptyset_i, s\mi))\big) \\
&\qquad \overset{(2)}{\leq}  w_i\big(C_i(\bar s_i, s\mi) + \sum_{j\neq i} (C_j(\bar s_i, s\mi) - C_j(\emptyset_i, s\mi))\big) \\
&\qquad \overset{(1)}{=}  w_i \sum_j (C_j(\bar s_i, s\mi) - C_j(\emptyset_i, s\mi))\\ 
&\qquad\overset{(3)}{\leq}  C(\bar s_i, s\mi) - C(\emptyset_i, s\mi) 
= \bar C_i (\bar s_i, s\mi).
\end{align*}
Summing over all $i$, it follows that if $\bar G$ is $(\lambda,\mu)$-smooth\footnote{in the sense that there exist $\bar s \in \Sigma$ and an optimal $s^* \in \Sigma$ such that for all $s \in \Sigma $ it holds that $\sum_i C_i(\bar s_i, s\mi) \leq \lambda C(s) + \mu C(s^*)$, generalizing Roughgarden's definition of smoothness \cite{roughgarden}.}, then so is $G^\alpha.$
\qed
\end{proof}

If all weights are 1, then SC-boundedness follows from strong SC-bounded\-ness. To see this, consider the case where $\alpha = \textbf{0}$ and carry out the proof of Theorem \ref{corresp_relationship} for $s$ instead of $(\bar s_i, s\mi).$

\section{Minsum Machine Scheduling}
A \emph{scheduling game} $G = (m,n,(p_{ij})_{i \in M, j \in N}, (w_j)_{j \in N})$ consists of a set of jobs (players) $[n] = \{1, \ldots, n\}$ and a set of machines $[m] = \{1, \ldots, m\}$. For each machine $i$ and job $j$, $p_{ij} \in \mathbb R_+$ denotes the \emph{processing time} of $j$ on $i$. Furthermore, $w_j$ is the \emph{weight} of job $j$. The strategy space $\Sigma_i$ of a job $j$ is simply the set of machines. By $\emptyset_i= \emptyset$ we mean the strategy where $i$ uses no machine.

Let $x$ be a strategy profile. For a machine $i$, we denote by $X_i$ the set of jobs that are scheduled on $i$. Furthermore, $x_j$ denotes the machine $j$ is assigned to. Following the notation by Cole et al. \cite{cole11}, we define $\rho_{ij} = \frac{p_{ij}}{w_j}$. We assume that the jobs on a machine are scheduled in increasing order of $\rho_{ij}$, which is known as \emph{Smith's rule} \cite{smith}; if two jobs on a machine have the same time-to-weight ratio, we use a tie-breaking rule. The \emph{cost} $C_j$ of job $j$ which it seeks to minimize is simply its completion time. In the following, we assume for simplicity that the $\rho_{ij}$ are pairwise distinct (but the results continue to hold without this assumption). Then we can write
\[
C_j(x) = \sum_{k \in X_i :\: \rho_{ik} \leq \rho_{ij}} p_{ik}.
\]
The social cost $C$ we consider is the weighted sum of the players' completion times, i.e., $C = \sum_j w_j C_j$.

In the following, we use the three-field notation by Graham et al \cite{graham}. In this notation, the problem we described is denoted by $R||\sum_j w_j C_j.$ If all weights are $1$, we write $\sum_j C_j$ instead of $\sum_j w_j C_j$. 
Furthermore, if there are \emph{speeds} $s_i$ for each machine $i$ and \emph{fixed} processing times $p_j$ for each job such that $p_{ij} = p_j / s_i$, we write $Q$ instead of $R$. Finally, if we have in addition identical speeds $s_i = 1$ for all machines $i$, the problem is denoted by $P$.

\subsection{$R || \sum_j w_j C_j$}
\begin{lemma}[\cite{cole11}]\label{lemma_cole}
For all strategy profiles $x$ and $x^*$,
\[
\sum_{i \in [m]} \sum_{j \in X_i^*} w_j p_{ij} + \sum_{i \in [m]} \sum_{j \in X_i^*} \sum_{k \in X_i} w_j w_k \min\{\rho_{ij}, \rho_{ik}\} \leq 2C(x^*) + \frac{1}{2} C(x),
\]
where $X_i^*$ is defined similarly to $X_i$ as $X_i^* = \{j \in J \st x^*_j = i\}.$
\end{lemma}
\begin{proof}
The claim is shown in the proof of \cite[Theorem 3.2]{cole11}.
\qed
\end{proof}

\begin{theorem}
Let $G$ be an instance of $R || \sum_j w_j C_j$ that satisfies the following condition for all jobs $j, k$ and all machines $i$: $\rho_{ij} \leq \rho_{ik}$ implies $w_j \leq w_{k}$ (i.e., if $k$ gets scheduled after $j$ on $i$, then it is because of its processing time, not its weight). Then the robust PoA of all friendship extensions $G^\alpha$ of $G$ is at most $4$.
\end{theorem}

For jobs $j$ and $k$, $\alpha_{jk}$ has an influence on $j$'s strategy in an equilibrium only if there is a machine $i$ such that $k$ gets scheduled after $j$ on $i$ because $j$ cannot influence $k$'s costs otherwise. Hence the weight condition tells us that the only jobs that could potentially have an influence on $j$ are in fact the jobs that are at least equally important as $j$. Hence $j$ cannot `misplace his affections' and care too much about unimportant jobs.

\begin{proof} 
First we show that $G$ is strongly SC-bounded. Clearly, (1) and (2) are satisfied. It remains to show that (3) holds. For all jobs $j$ and strategy profiles $x$,
\[
\sum_k (C_k(x) - C_k(\emptyset, x\mi)) = C_j(x) + \sum_{k :\: x_k = x_j, \, \rho_{x_jk}>\rho_{x_jj}} p_{x_jj}.\]
It follows that if $i = x_j$, then
\begin{eqnarray*}
w_j \sum_k (C_k(x) - C_k(\emptyset, x\mi)) &=&  w_j \Big(C_j(x) + \sum_{k \in X_i: \:\rho_{ik} > \rho_{ij}} p_{i j} \Big) \\
&\leq &   w_j C_j(x) + \sum_{k \in X_i: \:\rho_{ik} > \rho_{ij}} w_{k} p_{ij}  = \bar C_j(x),
\end{eqnarray*}
where the inequality follows from the condition on the weights. So $G$ is indeed strongly SC-bounded.

We calculate 
\begin{eqnarray*}
\bar C_j (x^*_j, x\mj) &=& w_j C_j(x_j^*, x\mj) + \sum_{k \in X_i: \:\rho_{ik} > \rho_{ij}} w_{k} p_{ij} \\
&=& w_j p_{ij} + \sum_{k \in X_i: \: \rho_{ik} < \rho_{ij}} w_k w_j \rho_{ik} + \sum_{k \in X_i: \:\rho_{ik} > \rho_{ij}} w_{k} w_j \rho_{ij} \\
&\leq & w_j p_{ij} + \sum_{k \in X_i} w_j w_k \min\{\rho_{ij}, \rho_{ik}\}.
\end{eqnarray*}
Summing over all machines $i$ and $j \in X_i^*$, this is the same expression as in Lemma \ref{lemma_cole}. 
Hence 
$\sum_j \bar C_j (x^*_j, x\mj) \leq 2C(x^*) + \frac{1}{2} C(x)$
and $\bar G$ is $(2, 1/2)$-smooth. It follows by Theorem \ref{corresp_relationship} that the robust PoA in the friendship model is at most $4$.
\qed 
\end{proof}

This bound is \textbf{tight}: \cite{correa12} shows that the pure PoA for $RP||\sum_jC_j$ is 4. $RP||\sum_j C_j$ is almost defined as $P||\sum_j C_j$ with the exception that each player $i$ can only use a \emph{subset} of the set of machines, i.e. $\Sigma_i \subseteq [m]$ ($R$ stands for \emph{restricted}). Consider an instance of $RP||\sum_jC_j$. We can simulate restrictions in the $R || \sum_jC_j$ setting by letting $p_{ij} > \max_{x} C_j(x)$ for machines $i$ that are not allowed for $j$, where $x$ runs over the feasible schedules of the original instance. Then $j$ neither chooses $i$ in a Nash equilibrium nor in the optimal schedule. Hence the PoA stays the same in the new game. Thus the lower bounds in \cite{correa12} also work for our setting.

The \textbf{weight condition} is \textbf{necessary}. In fact, if we drop it, the pure PoA is unbounded even for $P||\sum_j w_j C_j$ instances with unit-size jobs. An illustrating example is given in Appendix~\ref{app:weight-example}.

\subsection{$P || \sum_j C_j$}
Fix an ordering of the jobs such that $p_j>p_{j'}$ implies $j>j'$. We use the same notation as in  \cite{hoeksmauetz}: For a schedule $x$, a job $j$ and a machine $i$, let $h^x_i(j) = |\{j' > j \st x_{j'} = x_j\}|$. This is the number of jobs that are scheduled after $j$ on $i$. Using this notation, we can write $\bar C_j(x) = C_j(x) + h_{x_j}^x(j) \cdot p_j$ for instances with unit speeds.

Throughout this section, let $\bar x$ denote the randomized schedule that assigns each job to each machine with probability $\frac{1}{m}.$

\begin{lemma}\label{mixedcost}
Let $x$ be an arbitrary schedule.
Then 
$$\sum_j \E[C_j (\bar x_j, x_{-j})] = \frac{1}{m} \sum_j p_j (m+n-j).$$
\end{lemma}
Note that, surprisingly, this is independent of $x$.
\begin{proof} 
Clearly,
\[
\sum_j \E[C_j(\bar x_j, x_{-j})] = \sum_j \Big (p_j + \sum_{k < j}  p_k \cdot P(\bar x_j = x_{k})\Big ) = \sum_j \Big( p_j + \sum_{k<j} \frac{p_k}{m} \Big ).
\]
Reordering the second sum gives us
$$
\sum_j \E[C_j(\bar x_j, x_{-j})] 
= \sum_j p_j + \sum_j \sum_{k>j} \frac{p_j}{m} = \frac{1}{m} \sum_j p_j \left( m + n - j\right). 
\eqno\qed
$$
\end{proof}

The following theorem will be helpful to establish an upper bound on the robust PoA for the friendship model and might be of independent interest. We defer its proof to Appendix~\ref{app:MFT}.

\begin{theorem}\label{thm:rpoa-p} For any schedule $x$ and any optimal 
$x^*$, 
\[\sum_j C_j(\bar x_j, x\mj) \leq C(x^*) + \Big(\frac{1}{2}-\frac{1}{2m}\Big)\sum_j p_j.\]
In particular, the robust price of anarchy of $P || \sum_j C_j$ is at most $ \frac{3}{2} - \frac{1}{2m}$. This bound is tight.
\end{theorem}

\begin{theorem}
Let $G$ be an instance of $P||\sum_j C_j$. Then the robust PoA for any friendship extension $G^\alpha$ is at most $2$.
\end{theorem}
\begin{proof}
Let $x$ be arbitrary. Then by linearity of expectation,
\[
\E \Big[ \sum_j \bar C_j(\bar x_j, x\mj) \Big] = \sum_j \E[C_j(\bar x_j, x\mj)] + \sum_j \E[h^x_{\bar x_j}(j)] \cdot p_j.
\]
We know that
\[
\E[h^x_{\bar x_j}(j)] = \frac{1}{m} \sum_i h^x_i(j) = \frac{1}{m} |\{j' \in J \st j' > j\}| = \E[h^{\bar x}_{x_j}(j)].
\]
Hence the second term evaluates as
\[
\sum_j \E[h^x_{\bar x_j}(j)] \cdot p_j 
= \sum_j \E[h_{x_j}^{\bar x} (j)] \cdot p_j 
= \sum_j \E[C_j(\bar x_j, x_j)] -\sum_j p_j.
\]
We know by Theorem \ref{thm:rpoa-p} that $\sum_j \E[C_j(\bar x_j, x \mj)] \leq C(x^*) + (\frac{1}{2}-\frac{1}{2m}) \sum_j p_j$.
Hence 
\begin{eqnarray*}
\sum_j  \E[\bar C_j(\bar x_j, x\mj)] &=& 2 \sum_j \E[ C_j(\bar x_j, x \mj)] - \sum_j p_j \\
&\leq & 2C(x^*) -\frac{1}{m} \sum_j p_j  \leq  2 C(x^*),
\end{eqnarray*}
for any schedule $x^*$. Hence the robust PoA for the friendship extension is at most $2$.
\qed
\end{proof}

\section{Linear Congestion Games} 

An \emph{atomic congestion game} $G=(N, E, \{\Sigma_i\}_{i \in N}, (d_e)_{e \in E})$ is given by a set $E$ of \emph{resources} together with \emph{delay functions} $d_e: \mathbb N \to \mathbb R_+$ indicating the delay on $e$ for a given number of players using $e$. Each player's strategy set consists of subsets of $E$; $\Sigma_i \su \mathcal P(E)$ for all $i$. For $s \in \Sigma$, let $x_e(s) = |\{ i \in N \st e \in s_i\}|.$ The cost of each player $i$ under $s$ is given by $C_i(s) = \sum_{e \in s_i} d_e(x_e(s))$.
If all delay functions are linear, we say that $G$ is \emph{linear}. The social cost $C$ is simply the sum over all individual cost. By $\emptyset_i = \emptyset$ we mean the strategy where player $i$ uses no machine.

It is known that we can without loss of generality assume that all latency functions are of the form $l_e(x) = x$. This was first mentioned in \cite{christodoulou}. For a proof, see \cite{altruistic_games}. The following lemma is shown in the proof of \cite[Theorem 1]{christodoulou}.

\begin{lemma}[\cite{christodoulou}]\label{congestion_base_game}
Let $G$ be a linear congestion game and $s, s^* \in \Sigma.$ Then 
\[
\sum_i C_i(s_i^*, s\mi) \leq \sum_e x_e(s^*)(x_e(s) + 1).
\]
\end{lemma}

\begin{lemma}[\cite{bilo}]\label{lem:bilo}
For any pair $\alpha, \beta \in \mathbb N$, it holds that $\frac{2}{5} \alpha^3 + \frac{17}{5} \beta^2 \geq \beta(\alpha+1).$
\end{lemma}

Bil\`o et al. show in their paper \cite{bilo} that the \emph{pure} PoA lies between 5 and $17/3$ for a restricted friendship setting, where $\alpha_{ij} \in \{0,1\}$ for all $i,j$. We generalize their result to the \emph{robust} PoA for arbitrary $\alpha_{ij} \in [0,1]$ and show tightness.

\begin{theorem}
Let $G$ be a linear congestion game. Then the robust PoA of all friendship extensions $G^\alpha$ is bounded by $\frac{17}{3}\approx 5.67$. This bound is tight.
\end{theorem}
\begin{proof}
We have
\[
\bar C_i(s) = C_i(s) + \sum_{e \in s_i} |\{j \neq i \st e \in s_j\}|= C_i(s) + \sum_{e \in s_i} x_e(\emptyset, s\mi)  \geq C_i(s),
\]
so $G$ is SC-bounded. Also $G$ is strongly SC-bounded: If $i$ does not use any machine, he experiences no waiting time; the other's completion times can only increase if another player enters; and finally, $C = \sum_j C_j$.

Let $s, s^* \in \Sigma$. We abbreviate $x_e(s)$ and $x_e(s^*)$ by $x_e$ and $x_e^*$, respectively. The calculation of the robust PoA for $\bar G$ yields
\[
\sum_i \bar C_i(s_i^*, s\mi) = \sum_i C(s_i^*, s\mi) + \sum_i \sum_{e \in s_i^*} x_e(\emptyset, s \mi).
\]
The first term is at most $\sum_e x^*_e(x_e+1)$ by Lemma \ref{congestion_base_game}. The second term is bounded from above by $\sum_i \sum_{e \in s_i^*} x_e(s) = \sum_{e \in E} x_e x_e^*$. Hence we get in total by Lemma \ref{lem:bilo}
\[
\sum_e x_e^* (2x_e + 1)  \leq \sum_e \left ( \frac{17}{5} (x_e^*)^2 + \frac{2}{5} x_e^2 \right) = \frac{17}{5} C(s^*) + \frac{2}{5} C(s).
\]
It follows that the robust PoA of $\bar G$ is at most
$
\frac{17}{5}/(1-\frac{2}{5}) = \frac{17}{3}.
$

We show now that the bound of $\frac{17}{3}$ is asymptotically tight. Let $n \geq 0$. Consider an instance with $n+3$ blocks of players $B_0, \ldots, B_{n+2}$ consisting of three players each: $B_k = \{a_k, b_k, c_k\}$.  We construct a Nash equilibrium $s$ and an optimal strategy profile $s^*$ as follows. For all resources $e$, we set $l_e(x) = x$. For $0 \leq k \leq n$, the pattern of strategies repeats (see Figure 1). Here player $i = a_k$ has two strategies $s_i = \{3k, 3k+1, 3k+2\}$ and $s_i^* = \{3k + 6\}$. Player $i = b_k$ has two strategies $s_i = \{3k+2, 3k+3\}$ and $s_i^* = \{3k + 7\}$. Player $i = c_k$ has two strategies $s_i = \{3k+3, 3k+4\}$ and $s_i^* = \{3k + 8\}$. 

\input{Figure.tex}
The strategies $s_i$ of players in the final blocks $B_{n+1}$ and $B_{n+2}$ are defined as above. However, we need to change the definition of $s_i^*$ because otherwise, $s$ is not a Nash equilibrium. So for each $i \in B_{n+1} \cup B_{n+2}$, we insert sets of new, previously unused resources $s_i^*$ such that $C_i(s_i) = |s_i^*|$.

For the following tuples of players $(i,j)$ it holds that $\alpha_{ij} = 1$: $(a_k, b_{k+1}), (a_k, c_{k+1}), (a_k, a_{k+2})$ as well as $(b_k, c_{k+1}), (b_k, a_{k+2})$ and $(c_k, a_{k+2}), (c_k, b_{k+2})$, where $0 \leq k \leq n$. All other $\alpha_{ij}$ are zero. Hence $\alpha_{ij} = 1$ iff $s_i^*$ intersects $s_j$. Note that if $s_i \cap s_j \neq \emptyset$, then $\alpha_{ij} = 0.$

Now, we claim that $s$ is a Nash equilibrium. In fact, for all $0 \leq k \leq n$ and $i = a_k$, $C^\alpha(s) = C(s) + \sum_{j \neq i} \alpha_{ij} C_j(s) = 7 +  5+5+7 = 24$, which equals $C_i^\alpha(s^*, s\mi) = 4 + 6 + 6 + 8$. A similar calculation shows $C_i(s) = C_i(s_i, s\mi)$ for $i = b_k, c_k$. Observe that for $k = n+1, n+2$, and $i \in B_k$, $C_i^\alpha(s) = C(s) = |s^*_i| = C(s_i^*, s\mi)$ by our construction of $s_i^*$. Hence $s$ is indeed a Nash equilibrium. It is easy to see that $s^*$ is optimal.

For $k = 1, \ldots, n$, block $B_k$ has the same cost: $C(B_k) := \sum_{i \in B_k} C_i(s) = 17$ and $C^*(B_k) := \sum_{i \in B_k} C_i(s^*) = 3$. Let $X = C(B_0) + C(B_{n+1}) + C(B_{n+2})$ and $X^*  = C^*(B_0) + C^*(B_{n+1}) + C^*(B_{n+2})$ and observe that these are constants independent of $n$. It follows that
\[
\frac{C(s)}{C(s^*)} = \frac{17n + X}{3n + X^*} = \frac{17 + o(n)}{3 + o(n)}. \eqno{\qed}
\]
\end{proof}

\section{Auctions}
An auction $G$ consists of an \emph{allocation rule} $a: \Sigma \to N$ which determines which bidder gets the item and a \emph{pricing rule} $p: \Sigma \to \mathbb R^N$ indicating how much each player should pay. Each bidder $i$ is assumed to have a certain valuation $v_i \in \mathbb R_+$ for the item. For a given bidding profile $b \in \mathbb R_+^N$, the social welfare is 
$\Pi(b) = v_{a(b)}$. Player $i$'s utility is given by $\Pi_i(b) = v_i - p_{i}(b)$ if he gets the object and $-p_i(b)$ otherwise. In a \emph{second-price auction}, the highest bidder gets the item and pays the second highest bid, while everybody else pays nothing.

We do not allow \emph{overbidding}, i.e., for all bidders $i$, $b_i \leq v_i$. This is a standard assumption because overbidding is a dominated strategy.  We denote by $\beta(b,i)$ the name of the player who places the $i$-th highest bid in $b$. We write $\beta(i)$ instead of $\beta(b,i)$ if the bidding profile is clear from the context. $\emptyset_i = 0$ denotes the strategy where bidder $i$ bids nothing.

\begin{theorem}
Consider an auction $G$ with an allocation rule as in the second-price auction and a pricing rule $p$ where every bidder pays at most what he would pay in a second-price auction (for every given bidding profile). 
Then the robust PoA of all altruism extensions $G^\alpha$ is at most $2$.
\end{theorem}
\begin{proof}
Let $G$ be an auction of the described type. We show that $G$ is SC-bounded. Let $b$ be a bidding profile. We calculate
\[
\bar \Pi_i(b) = \Pi(b) - \Pi(0, b\mi) = 
\begin{cases}
v_{\beta(1)}-v_{\beta(2)}, & i = \beta(1) \\
0, & \text{otherwise}.
\end{cases}
\]
Note that $\Pi_{\beta(1)} \geq v_{\beta(1)} - b_{\beta(2)} \geq v_{\beta(1)} - v_{\beta(2)}$ because we do not allow overbidding. Hence $\Pi_i(b) \geq \bar \Pi_i(b)$ for all $i$. So $G$ is SC-bounded.

Now, let $b^*$ be the optimal bidding profile where the bidder with the highest valuation, say bidder 1, bids his valuation and everybody else bids nothing. Let $b$ be arbitrary. Then
$
\sum_i \bar \Pi_i(b^*, b\mi) = \sum_i (\Pi(b_i^*, b\mi) -  \Pi(0, b\mi) )= \Pi(v_1, b_{-1}) -  \Pi(0, b_{-1}).
$
Now, we distinguish two cases: Either bidder 1 wins under $b$ and then $b$ is optimal and $\Pi(0, b_{-1}) \leq \Pi(b)$. Otherwise, the winner remains the same if 1 does not bid anything, so $\Pi(0, b_{-1}) = \Pi(b)$. In any case, $\Pi(0, b_{-1}) \leq \Pi(b)$. Furthermore, $\Pi(v_1, b_{-1}) = v_1 = \Pi(b^*)$ because no bidder can overbid $v_1$. It follows that the term above is bounded from below by $\Pi(b^*) - \Pi(b)$. 
Thus the robust PoA is at most $2$.
\qed
\end{proof}

\begin{theorem}\label{thm:auction_friendship}
Let $G$ be a second-price auction. Then the coarse PoA of the class of friendship extensions of $G$ is exactly $2$.
\end{theorem}

Note that here the friendship model is \emph{not} a generalization of the altruism model because $\Pi \neq \sum_i \Pi_i$. We defer the proof to Appendix \ref{app:auctions}.

\section{Valid Utility Games}\label{section_valid_utility}

\sloppy

A \emph{valid utility game} \cite{vetta} is defined as a payoff-maximization game $G = (N,E,\{\Sigma_i\}_{i \in N}, \{\Pi\}_{i \in N}, V)$, where $E$ is a ground set of resources, $\Sigma_i \subseteq \mathcal P(E)$ and $V$ is a submodular and non-negative function on $E$. The social welfare $\Pi$ is given by $\Pi(s) = V(\bigcup_{i \in N} s_i)$ and is assumed to be sum-bounded. Furthermore, we require $G$ to satisfy $\Pi_i(s) \geq \Pi(s) - \Pi(\emptyset, s\mi)$
for all $s \in \Sigma$. If $G$ additionally satisfies the last inequation with equality, it is called \emph{basic utility game} \cite{vetta}. For all players $i$, set $\emptyset_i = \emptyset$.

\begin{theorem}[\cite{roughgarden}]
The robust PoA of valid utility games with non-decreasing\footnote{where non-decreasing means that for all $A\su B \su E$ it holds that $V(A) \leq V(B)$.} set function $V$  is bounded by $2$.
\end{theorem}

An example for valid utility games with non-decreasing set functions are competitive facility location games without fixed costs \cite{vetta}. 

The following theorem has already been proven in \cite{altruistic_games} and tightness of this bound has been shown in \cite{blum} for the base game. We now use our framework to provide a shorter proof that illustrates nicely why the robust PoA does not increase for altruistic extensions: The corresponding SCG falls into the same category of games.
\begin{theorem}
Let $G$ be a valid utility game with non-decreasing $V$.  Then the robust price of anarchy of every altruistic extension $G^\alpha$ of $G$ is bounded by $2$. 
\end{theorem}
\begin{proof}
It follows directly from the definition that $G$ is SC-bounded. It is easy to verify that the corresponding SCG $\bar G = (N,E,\{\Sigma_i\}_{i \in N}, \{\bar \Pi\}_{i \in N}, V)$ is again a valid utility game: $\sum_i \bar \Pi_i (s) \leq \sum_i \Pi_i(s) \leq \Pi(s)$ and $\bar \Pi_i(s) = \Pi(s) - \Pi(\emptyset, s\mi)$. So the robust PoA of $\bar G$ is at most $2$. Our claim follows by Theorem \ref{corresp_altruism}.
\qed
\end{proof}

\bibliographystyle{splncs03}
\bibliography{Paper} 

\newpage
\appendix

\section{Proof of Theorem \ref{weighted_smoothness_thm}} \label{weighted_smoothness_appendix}
\textbf{Theorem \ref{weighted_smoothness_thm}.}~
\emph{Let $G^\alpha$ be a friendship extension of a cost-minimization game with weight-bounded social cost function $C$. If $G^\alpha$ is $(\lambda, \mu)$-smooth with $\mu < 1$, then the coarse PoA of $G^\alpha$ is at most $\frac{\lambda}{1-\mu}$.}
\begin{proof}
The proof works similarly to \cite[Theorem 1]{inefficiency}:

Let $\sigma$ be a coarse equilibrium for $G^\alpha$ and let $s$ be a random variable with distribution $\sigma$. In addition, let $s^*$ be an arbitrary strategy profile and let $\bar s$ be as in the definition of smoothness. We assume without loss of generality that $\bar s$ is a pure strategy profile; the arguments also work in the mixed case. Because $\sigma$ is a coarse equilibrium, for all players $i$ we have
\[
\mathbf E \Big [C_i(s) + \sum_{j \neq i} \alpha_{ij} C_j(s) \Big ] \leq \mathbf E \Big [C_i(\bar s_i, s\mi) + \sum_{j \neq i} \alpha_{ij} C_j(\bar s_i, s\mi) \Big]
\]
Using linearity of expectation, it follows that
\begin{eqnarray*}
\mathbf E [C(s)] 
&\leq& \sum_i w_i \mathbf E[C_i(s)] \\
&\leq & \sum_i w_i \mathbf E \Big[C_i(\bar s_i, s\mi) + \sum_{j \neq i} \alpha_{ij} (C_j(\bar s_i, s\mi)-C_j(s)) \Big] \\
&=& \mathbf E \Big[\sum_i w_i \Big( C_i(\bar s_i, s\mi) + \sum_{j \neq i} \alpha_{ij} (C_j(\bar s_i, s\mi)-C_j(s)) \Big) \Big] \\
&\leq & \mathbf E [\lambda C(s^*) + \mu C(s)] \\ 
&=&  \lambda C(s^*) + \mu \mathbf E [C(s)]
\end{eqnarray*}
Hence $\mathbf E[C(s)] \leq \frac{\lambda}{1-\mu} C(s^*).$
\qed
\end{proof}

\section{Necessity of Weight Condition}
\label{app:weight-example}

Let us assume we have $m$ machines and $m$ jobs of weight $1$ as well as $m(m-1)$ jobs of weight 0. Let $A_i$ ($i = 1,2$) denote the set of jobs of weight $i$. Set $\alpha_{jk} = 1$ if $j \in A_1, k \in A_0$, $0$ otherwise.

First, consider the schedule $x$ where every job in $A_1$ gets scheduled on machine $1$ and all the jobs from $A_0$ are distributed among the remaining $m-1$ machines such that every machine $i \in \{2,\ldots, m\}$ gets exactly $m$ jobs. We can assume that the tie-breaking rule among jobs in $A_0$ is such that they cannot improve their completion time by deviating. Then $x$ is a Nash equilibrium: Indeed, let $j \in A_1$. Then $C_j(x) \leq m$ and for all $i \in \{2, \ldots, m\}$, $C_j(i, x\mj) = 1 + \sum_{k \in A_0: \: x_k = i} 1 = 1+m$. Hence $j$ has no incentive to deviate. Note that $C(x) = \sum_{j = 1}^m j = \frac{1}{2}m(m+1).$

Now, in an optimal schedule $x^*$, the jobs are distributed among the machines in such a way that every machine completes exactly one job of weight 1. Hence an optimal schedule satisfies $C(x^*) = m$.

It follows that the pure PoA is at least $\frac{1}{2} (m+1)$ and thus unbounded.

\section{Robust PoA of $P||\sum_j C_j$}
\label{app:MFT}

In order to characterize the optimal solution, we use the \emph{Minimum Mean Flow Time (MFT)} algorithm \cite{horowitz} that produces an optimal schedule for $Q|| \sum_j C_j$. A formal description of the algorithm is given below \cite{hoeksmauetz}.

\begin{algorithm}[H]
\caption{The MFT Algorithm}
For each machine $i$ set $h_i = 0$
\begin{algorithmic}
\WHILE{not all jobs are placed}
\STATE Take the longest job $j$ of the set of unscheduled jobs
\STATE Assign $j$ to the machine $i$ with the smallest value of $(h_i +1) /s_i$
\STATE For the chosen machine update $h_i := h_i +1$
\ENDWHILE
\STATE Sort the jobs on each machine in SPT order
\end{algorithmic}
\end{algorithm}

\sloppy
\begin{lemma}\label{optimalcost}
Let $x^*$ be an optimal schedule for $P || \sum_j C_j$. Then 
$C(x^*) = \sum_j p_j  ( 1+  \lfloor (n-j)/m \rfloor )$.
\end{lemma}
\begin{proof}
Without loss of generality, we can assume that $x^*$ is generated by the MFT algorithm.
Consider some job $j$ and let $i = x^*_j$. For each job $j'$ that is considered after $j$ (i.e., each job with smaller index), the algorithm chooses a machine $i'$ that minimizes $h_{i'}^{x^*}(j')$. So it chooses exactly $m-1$ other machines before it places another job on $i$ (provided that the algorithm always uses the same tie-breaking rule on the set of machines). Hence $j$ causes a delay of $p_j$ for himself and for $\left\lfloor (n-j)/m \right\rfloor$ other machines. Summing over all jobs $j$, the formula follows.
\qed
\end{proof}

\begin{lemma}\label{linear_weights}
Let $p_1 \leq \ldots \leq p_m$ be a sequence of reals. Then
$$
\left ( \dfrac{1}{2}- \dfrac{1}{2m} \right ) \sum_{j = 1}^m p_j \geq \sum_{j = 1}^m \dfrac{m-j}{m} p_j.
$$
\end{lemma}

\begin{proof}
Let $\delta_j = p_j - p_{j-1} \geq 0$ for $j \in \{1, \ldots, n\}$, where $p_{0}$ is set to $0$. It holds that
\begin{eqnarray}\label{rechnung1}
m \sum_j \sum_{k \leq j} \delta_j &=& \sum_j (m-j+1)m \, \delta_j \nonumber \\
&\geq& \sum_j (m-j+1)(1+(m-j)) \delta_j \nonumber \\
& = & \sum_j (m-j+1) \delta_j + \sum_j (m-j)(m-j+1) \delta_j \nonumber \\
& = & \sum_j \sum_{k \leq j} \delta_k + \sum_j (m-j)(m-j+1) \delta_j
\end{eqnarray}
The second sum equals
\begin{eqnarray}\label{rechnung2}
 2 \sum_k \frac{(m-k)(m-k+1)}{2} \delta_k & =&  2 \sum_k \left( \delta_k \sum_{j=1}^{m-k} j \right) \nonumber \\
 &=& 2 \sum_j (m-j) \sum_{k \leq j} \delta_k
\end{eqnarray}
Combining (\ref{rechnung1}) and (\ref{rechnung2}), we get
\begin{eqnarray*}
m \sum_j p_j &=& m \sum_j \sum_{k \leq j} \delta_j \nonumber \\
&\geq& \sum_j \sum_{k \leq j} \delta_k + 2 \sum_j (m-j) \sum_{k \leq j} \delta_k \nonumber \\
& = &  \sum_j p_j + 2 \sum_j (m-j) p_j.
\end{eqnarray*}
Thus $(m-1) \sum_j p_j \geq  2 \sum_j (m-j) p_j$. Multiplying by $\dfrac{1}{2m}$, the claim follows.
\qed
\end{proof}

Recall that by $\bar x$ we denote the mixed schedule that assigns each job to each machine with equal probability.\newline

\noindent
\textbf{Theorem~\ref{thm:rpoa-p}.} 
\emph{For any schedule $x$ and any optimal 
$x^*$, 
\[\sum_j C_j(\bar x_j, x\mj) \leq C(x^*) + \Big(\frac{1}{2}-\frac{1}{2m}\Big)\sum_j p_j.\]
In particular, the robust price of anarchy of $P || \sum_j C_j$ is at most $ \frac{3}{2} - \frac{1}{2m}$. This bound is tight.} 
\begin{proof}
First we bound the robust PoA from above.

Let $x$ be an arbitrary schedule and suppose $x^*$ is optimal. Then by Lemma \ref{mixedcost},
\begin{eqnarray*}
\sum_j \E[ C_j(\bar x_j, x_{-j})] &=& \sum_j p_j \left(1+\frac{n-j}{m}\right) \\
&=& \sum_j p_j \left(1+ \left\lfloor \frac{n-j}{m} \right\rfloor\right) \
+ \sum_j p_j \left( \frac{(n-j) \mod m}{m} \right)
\end{eqnarray*}
The left sum evaluates as $C(x^*)$ by Lemma \ref{optimalcost}. It follows from Lemma \ref{linear_weights} that the right sum is at most $\left( \frac{1}{2} - \frac{1}{2m} \right) \sum_j p_j$,
which in turn is at most $\left( \frac{1}{2} - \frac{1}{2m} \right) C(x^*)$. This shows the claim.

Now we give a lower bound on the robust PoA. Take $m = n$ and let all jobs have the same processing time, say $1$. Again, let $\bar x$ be the mixed schedule that assigns each job to each machine with probability $\frac{1}{m}$. The optimal schedule $x^*$ assigns exactly one job to each machine and thus has a cost of $m$. By Lemma \ref{mixedcost}, the cost of $\bar x$ evaluates as
\begin{eqnarray*}
C(\bar x) &= &  \sum_j \Big (1+\frac{n-j}{m} \Big) =  m + \frac{1}{m}\sum_j (m-j) = m +  \frac{m (m-1)}{2m} = \dfrac{3}{2}m - \frac{1}{2}.
\end{eqnarray*}
Hence the mixed price of anarchy (which is a lower bound for the robust PoA) is at least 
$\frac{3}{2}-\frac{1}{2m}$.
\qed
\end{proof}

\section{Proof of Theorem \ref{thm:auction_friendship}}\label{app:auctions}
\textbf{Theorem \ref{thm:auction_friendship}.}~
\emph{Let $G$ be a second-price auction. Then the coarse PoA of the class of friendship extensions of $G$ is exactly $2$.}

\begin{proof}
Unfortunately, $G$ is not strongly SC-bounded because assumption (3) is not satisfied. However, we can still bound the coarse PoA for $G^\alpha$ by using $\bar G$ in the following way.

Consider a Nash equilibrium $b$ and a valuation profile $v$ such that, say, bidder 1 has the highest value for the item. Let $b^*$ be as in the last proof, i.e., bidder 1 bids his value and everybody else bids nothing. If $\beta(b,1) = 1$, then $\Pi(b) = \Pi_1(b^*_1, b_{-1}) \geq \bar \Pi_1(b^*_1, b_{-1}) = \sum_i \bar \Pi_i(b_i^*, b\mi)$, so we can use the robust PoA of $\bar G$. Now, assume $\beta(b,1) \neq 1$. Then
\begin{eqnarray*}
0 = \Pi_1(b) &\geq & \Pi_1(b_1^*, b_{-1}) + \sum_{j \neq 1} \alpha_{1j} (\Pi_j(b_1^*, b_{-1}) - \Pi_j(b)) \\
&=& \Pi(b_1^*, b_{-1})- b_{\beta(b,1)} -  \alpha_{1 \beta(b,1)} \Pi_{\beta(b,1)}(b).
\end{eqnarray*}
We know that $b_{\beta(b,1)} \leq v_{\beta(b,1)} = \Pi(b) = \Pi(0, b_{-1})$. Also, $\alpha_{1 \beta(b,1)} \Pi_{\beta(b,1)}(b) \leq \Pi(b).$ Hence
\[
0 \geq \Pi(b_1^*, b_{-1})-(0, b_{-1}) - \Pi(b).
\]
Hence again $\Pi(b)$ is at least $\bar \Pi (b_1^*, b_{-1}) = \sum_i \bar \Pi_i(b^*_i, b\mi)$.  A canonical calculation shows that the same holds for coarse equilibria. Now, in the previous proof we saw that $\bar G$ is $(1, -1)$-smooth with respect to $b^*$. So the coarse PoA of $G$ is at most $2$.

It remains to show that this bound is tight. Consider the following situation: We have two bidders with $v_1 = 1$, $v_2 = 2$, $\alpha_{12} = \alpha_{21} = 1$. Clearly, it is optimal to allocate the item to bidder $2$ with a social welfare of $2$. However, the bidding profile $b = (1,0)$ is a Nash equilibrium. Indeed, bidder $2$ has a utility of $\Pi^\alpha(b) = 1$ which remains the same if he outbids player $1$.
\qed
\end{proof}
\end{document}

%% file: abstract.tex
\begin{abstract}
We introduce a new class of games, called \emph{social contribution games (SCGs)}, where each player's individual cost is equal to the cost he induces on society because of his presence. Our results reveal that SCGs constitute useful abstractions of altruistic games when it comes to the analysis of the robust price of anarchy. We first show that SCGs are \emph{altruism-independently smooth}, i.e., the robust price of anarchy of these games remains the same under arbitrary altruistic extensions.
We then devise a general reduction technique that enables us to reduce the problem of establishing smoothness for an altruistic extension of a base game to a corresponding SCG. Our reduction applies whenever the base game relates to a canonical SCG by satisfying a simple \emph{social contribution boundedness} property. As it turns out, several well-known games satisfy this property and are thus amenable to our reduction technique. Examples include min-sum scheduling games, congestion games, second price auctions and valid utility games. \mbox{Using} our technique, we derive mostly tight bounds on the robust price of anarchy of their altruistic extensions. For the majority of the mentioned game classes, the results extend to the more differentiated friendship setting. As we show, our reduction technique covers this model if the base game satisfies three additional natural properties. 
\end{abstract}

%% file: intro.tex
\paragraph{Motivation and Background.}

The study of the inefficiency of equilibria in strategic games has been one of main research streams in algorithmic game theory in the last decade and contributed to the explanation of several phenomena observed in real life. More recently, researchers have also started to incorporate more complex social relationships among the players in such studies, accounting for the fact that players cannot always be regarded as isolated entities that merely act on their own behalf (see also \cite{FS06}). In particular, the extent by which other-regarding preferences such as \emph{altruism} and \emph{spite} impact the inefficiency of equilibria has been studied intensively; see, e.g.,  \cite{BSS07,buehler,caragiannis,chenkempe,altruistic_games,networkdesign,skopalik,hoefer,Hoefer:2012,inefficiency}. 

In this context, some counterintuitive results have been shown that are still not well-understood. For example, in a series of papers \cite{buehler,caragiannis,altruistic_games} it was observed that for congestion games the inefficiency of equilibria gets worse as players become more altruistically, therefore suggesting that altruistic behavior can actually be harmful for society. On the other hand, valid utility games turn out to be unaffected by altruism as their inefficiency remains unaltered under altruistic behavior \cite{altruistic_games}. 
These discrepancies triggered our interest in the research conducted in this paper. The basic question that we are asking here is: What is it that impacts the inefficiency of equilibria of games with altruistic players?

To this aim, we consider two different models that have previously been studied in the literature: 
the \emph{altruism model} \cite{altruistic_games} and the \emph{friendship model} \cite{inefficiency}.  
In both models, one starts from a strategic game (called the \emph{base game}) specifying the \emph{direct cost} of each player and then extends this game by defining the \emph{perceived cost} of each player as a function of his neighbors' direct costs. 
In the {altruism model}, player $i$'s perceived cost is a convex combination of his direct cost and the overall social cost. In the more general {friendship model}, player $i$'s perceived cost is a linear combination of his direct cost and his friends' costs. 

In order to quantify the inefficiency of equilibria in our games we resort to the concept of the \emph{price of anarchy (PoA)} \cite{koutsoupias2}, which is defined as the worst-case relative gap between the cost of a Nash equilibrium and a social optimum (over all instances of the game). By now, a standard approach to prove upper bounds on the price of anarchy is through the use of the \emph{smoothness framework} introduced by Roughgarden \cite{roughgarden}. Basically, this framework allows us to derive bounds on the \emph{robust price of anarchy} by showing that the underlying game satisfies a certain \emph{$(\lambda, \mu)$-smoothness property} for some parameters $\lambda$ and $\mu$. The robust price of anarchy holds for various solution concepts, ranging from pure Nash equilibria to coarse correlated equilibria (see, e.g., Young~\cite{young}). 

The original smoothness framework \cite{roughgarden} has been extended to both the altruism and the friendship model in  \cite{altruistic_games} and \cite{inefficiency}, respectively. Applying these adapted smoothness frameworks to bound the robust price of anarchy is often technically involved because of the altruistic terms that need to be taken into account additionally (see also the analy\-ses in \cite{altruistic_games,inefficiency}). 

Instead, we take a different approach here. As we will show, there is a natural class of games, which we term \emph{social contribution games (SCGs)}, that is intimately connected with our altruism and friendship games. We establish a general reduction technique that enables us to reduce the problem of establishing smoothness for our altruism or friendship game to the problem of proving smoothness for a corresponding SCG. The latter is usually much simpler than proving smoothness for the altruism or friendship game directly. This also opens up the possibility to derive better bounds on the robust price of anarchy of these games through the usage of our new reduction technique.

\paragraph{Our Contributions.}

The main contributions presented in this paper are as follows:

\begin{itemize}
\item We introduce a new class of games, which we term \emph{social contribution games (SCGs)}, where each player's individual cost is defined as the cost he incurs on society because of his presence. Said differently, player $i$'s cost is equal to the difference in social cost if player $i$ is present/absent in the game. 
We show that SCGs are \emph{altruism-independently smooth}, i.e., if the SCG is $(\lambda, \mu)$-smooth then every altruistic extension is $(\lambda, \mu)$-smooth as well. 

\item We derive a general reduction technique to bound the robust price of anarchy of  both altruism and friendship games. Basically, the reduction can be applied whenever the underlying base game is \emph{social contribution bounded}, meaning that the direct cost of each player is bounded by his respective cost in the corresponding SCG (for the friendship model a slightly stronger condition needs to hold). It is worth mentioning that this reduction preserves the $(\lambda, \mu)$-smoothness parameters, i.e., the altruism or friendship game inherits the $(\lambda, \mu)$-smoothness parameters of the SCG. 

\item We generalize smoothness for friendship extensions to  \emph{weight-bounded} social cost functions. In previous papers, the used techniques usually required sum-bounded\-ness, which is a stronger condition \cite{inefficiency}. Applying this definition to scheduling games with weighted sum as social cost, we derive a nice characterization of those scheduling games whose robust PoA does not grow for friendship extensions.

\item We show that social contribution boundedness is satisfied by several well-known games, like min-sum scheduling games, congestion games, second-price auctions and valid utility games. Using our reduction technique, we then derive upper bounds on the robust price of anarchy of their friendship/altruism extensions. In most cases we prove matching lower bounds. The results are summarized in Table~\ref{tab:results}. 
\end{itemize}

\begin{table}[t]
\caption{Robust price of anarchy bounds derived in this paper for the friendship model (unless specified otherwise).  $\quad$ ($*$) \textit{ Our result only holds if a certain weight condition is satisfied. The previous best result was shown only for the special case $R || \sum_j C_j$.}}
\label{tab:results}
\begin{tabular}{|@{\;}l@{\;}|@{\;}c@{\;\;}r@{\;}|@{\;}p{4.6cm}@{\;}|}
\hline 
 & \multicolumn{2}{c@{\;}|@{\;}}{\textbf{Robust PoA}} & \textbf{Remarks} \\
 \textbf{Games} & our results & previous best &  \\
 \hline
 \hline
 $R || \sum_j w_j C_j \qquad \quad(*)$ & $= \, 4$ & $\leq 23.31$ \cite{inefficiency} \ & RPoA $= 4$ (selfish players) \cite{correa12} \\
 $P || \sum_j C_j$ & $\leq \, 2$ &  & RPoA $= \frac32 - \frac{1}{2m}$ (selfish players) \\
 linear congestion games & $= \frac{17}{3}$ & $\leq 7$ \cite{inefficiency} \ & $5 \leq$ PoA $\leq \frac{17}{3}$ (in special case) \cite{bilo} \\
 second price auctions & $= \,2$ &   &\\
 valid utility games & $= \,2$ & $= 2$ \cite{altruistic_games} \ & altruism model; RPoA $= 2$ (selfish players) \cite{roughgarden}  \\
\hline
\end{tabular}
\end{table}

\paragraph{Related Work.}~

Several articles propose models of altruism and spite \cite{bilo,BSS07,buehler,caragiannis,chenkempe,altruistic_games,networkdesign,skopalik,hoefer,Hoefer:2012}. 
Among these articles, the inefficiency of equilibria in the presence of altruism and spite was studied for various games in \cite{bilo,buehler,caragiannis,chenkempe,altruistic_games,networkdesign}. After its introduction in \cite{roughgarden}, the smoothness framework has been extended to incomplete information settings \cite{Roughgarden:2012,syrgkanis12a,syrgkanis12b} and altruism/spite settings \cite{altruistic_games,inefficiency}. 

The robust price of anarchy for minsum scheduling (not taking altruism or friendship into account) was studied in various papers such as \cite{hoeksmauetz}. They show that it does not exceed $2$ for $Q||\sum_j C_j$ (which we improve to a tight value of $3/2$ in the special case $P||\sum_jC_j$). A value of 4 for $R || \sum_j w_j C_j$ has been proven in \cite{cole11}.

Our work on linear congestion games generalizes a result in \cite{bilo}. They show that the pure price of anarchy does not exceed $17/3$ in a restricted friendship setting ($\alpha_{ij} \in \{0,1\}$). 

As indicated above, most related to our work are the articles \cite{altruistic_games,inefficiency}. We significantly improve the bounds on the robust price of anarchy for congestion games and unrelated machine scheduling games in \cite{inefficiency} and at the same time simplify the analysis by using our reduction technique.

%% file: Figure.tex
\begin{figure}
 \centering

\begin{tikzpicture}[scale=0.5]

\draw (-3,0.5) -- (19,0.5);
\draw (-3,-2.5) -- (19, -2.5);
\draw (-3,-5.5) -- (19, -5.5);
\draw (-3,-8.5) -- (19, -8.5);
\draw (-3,-11.5) -- (19, -11.5);

\fill[lightgray] (1,0) ellipse (1.5cm and 0.3cm);
\fill[black] (0,0) circle (0.15);
\fill[black] (1,0) circle (0.15);
\fill[black] (2,0) circle (0.15);

\fill[lightgray] (2.5,-1) ellipse (1cm and 0.3cm);
\fill[black] (2,-1) circle (0.15);
\fill[black] (3,-1) circle (0.15);

\fill[lightgray] (3.5,-2) ellipse (1cm and 0.3cm);
\fill[black] (3,-2) circle (0.15);
\fill[black] (4,-2) circle (0.15);

\fill[lightgray] (4,-3) ellipse (1.5cm and 0.3cm);
\fill[black] (3,-3) circle (0.15);
\fill[black] (4,-3) circle (0.15);
\fill[black] (5,-3) circle (0.15);

\fill[lightgray] (5.5,-4) ellipse (1cm and 0.3cm);
\fill[black] (5,-4) circle (0.15);
\fill[black] (6,-4) circle (0.15);

\fill[lightgray] (6.5,-5) ellipse (1cm and 0.3cm);
\fill[black] (6,-5) circle (0.15);
\fill[black] (7,-5) circle (0.15);

\fill[lightgray] (4+3,-3-3) ellipse (1.5cm and 0.3cm);
\fill[black] (3+3,-3-3) circle (0.15);
\fill[black] (4+3,-3-3) circle (0.15);
\fill[black] (5+3,-3-3) circle (0.15);

\fill[lightgray] (5.5+3,-4-3) ellipse (1cm and 0.3cm);
\fill[black] (5+3,-4-3) circle (0.15);
\fill[black] (6+3,-4-3) circle (0.15);

\fill[lightgray] (6.5+3,-5-3) ellipse (1cm and 0.3cm);
\fill[black] (6+3,-5-3) circle (0.15);
\fill[black] (7+3,-5-3) circle (0.15);

\fill[lightgray] (4+6,-3-6) ellipse (1.5cm and 0.3cm);
\fill[black] (3+6,-3-6) circle (0.15);
\fill[black] (4+6,-3-6) circle (0.15);
\fill[black] (5+6,-3-6) circle (0.15);

\fill[lightgray] (5.5+6,-4-6) ellipse (1cm and 0.3cm);
\fill[black] (5+6,-4-6) circle (0.15);
\fill[black] (6+6,-4-6) circle (0.15);

\fill[lightgray] (6.5+6,-5-6) ellipse (1cm and 0.3cm);
\fill[black] (6+6,-5-6) circle (0.15);
\fill[black] (7+6,-5-6) circle (0.15);

\draw (6,0) circle (0.15);
\draw (7,-1) circle (0.15);
\draw (8,-2) circle (0.15);
\draw (9,-3) circle (0.15);
\draw (10,-4) circle (0.15);
\draw (11,-5) circle (0.15);

\draw (12,-6) circle (0.15);
\draw (13,-7) circle (0.15);
\draw (14,-8) circle (0.15);
\draw (15,-9) circle (0.15);
\draw (16,-10) circle (0.15);
\draw (17,-11) circle (0.15);

\draw (-2,-1) node{$B_{k-1}$};
\draw (-2,-4) node{$B_{k}$};
\draw (-2,-7) node{$B_{k+1}$};
\draw (-2,-10) node{$B_{k+2}$};

\draw (0,-3) node{$a_k$};
\draw (0,-4) node{$b_k$};
\draw (0,-5) node{$c_k$};

\end{tikzpicture}
\label{figurname}
\caption{The strategy profiles $s$ (grey) and $s^*$ (white). Columns correspond to resources.}
\end{figure}
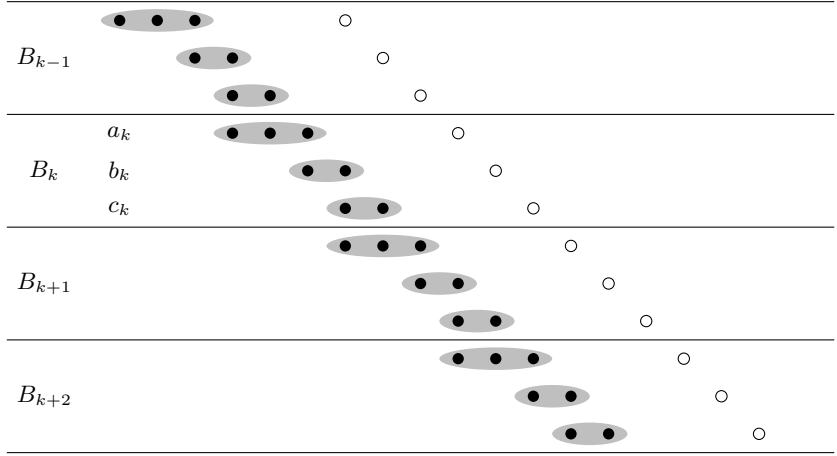